\documentclass[11pt]{article}

\usepackage{amsfonts}
\usepackage{mathrsfs}
\usepackage{amssymb}
\usepackage{stmaryrd}
\usepackage{amsmath}
\usepackage{graphicx}
\usepackage{subfig}
\usepackage{color}
\usepackage{ifpdf}
\usepackage{multirow}
\usepackage{arydshln}

\newtheorem{assumption}{Assumption}
\newtheorem{remark}{Remark}
\newtheorem{theorem}{Theorem}

\newtheorem{corollary}{Corollary}
\newtheorem{definition}{Definition}

\newenvironment{proof}{\makebox[7ex][l]{\it Proof:\/}}{\hfill\/ \hfill\/ {\it
Q.E.D.} \vspace{0.5ex}\\}

\begin{document}

\title{\LARGE \bf
IQC-Based Output-Feedback Control of LPV Systems \\
with Time-Varying Input Delays
}

\author{Fen Wu\thanks{Email: {\tt fwu@ncsu.edu}. Phone: (919) 515-5268.}
\vspace*{0.1in} \\
Department of Mechanical and Aerospace Engineering \\
North Carolina State University \\
Raleigh, NC 27695, USA}

\date{}

\maketitle

\begin{abstract}
Input delays are a common source of performance degradation and instability in control systems.
This paper addresses the $\mathcal{H}_\infty$ output-feedback control problem for LPV systems with time-varying input delays under the integral quadratic constraint (IQC) framework.
By integrating parameter-dependent Lyapunov functions with dynamic IQC multipliers, we derive convex, delay-dependent synthesis conditions formulated as parameter-dependent LMIs, enabled by the proposed exact-memory controller structure.
An explicit controller reconstruction formula is provided to recover the LPV controller from the LMI solution, avoiding the need to specify the functional form of the parameter-dependent controller gains.
While the synthesis problem for memoryless control is inherently non-convex, the proposed approach demonstrates significant performance improvement, reduced conservatism, and computational efficiency for standard output-feedback design.
Numerical examples illustrate the effectiveness and broad applicability of the method to LPV systems with time-varying input delays.
\end{abstract}


\section{Introduction}
\label{Sec.Intro}

Time delay is one of the most common nonlinear phenomena encountered in a wide range of engineering systems, including networked control systems, process control, and mechanical systems. The presence of delays often degrades system performance and may even lead to instability if not properly addressed. As a result, stability analysis and control synthesis for time-delay systems have attracted significant attention in the control community over the past several decades. Extensive results have been reported in the literature for both analysis and control design of time-delay systems. Various stability and performance criteria have been developed using different mathematical tools \cite{Fri.TAC06,KaoR.Au07}, while numerous control synthesis approaches have been proposed to address delay-dependent stabilization problems under different system configurations \cite{OliG.Au04,LeeMKP.Au04,YueH.Au05,Krs.CSM10}. Comprehensive surveys of this research area can be found in \cite{Ric.Au03,NicG.B04,XuL.IJSS08,Fri.B14}.

Existing approaches for time-delay systems can generally be classified into two main categories: frequency-domain methods and time-domain Lyapunov-based approaches. The frequency-domain approach determines stability by examining whether the roots of the associated characteristic equation lie in the left half-plane \cite{NicG.B04}. However, this approach is typically limited to systems with constant delays and is often complicated by the transcendental nature of the characteristic equation, which possesses infinitely many roots.  Most existing results rely on time-domain methods based on Lyapunov-Razumikhin functions (LRFs) or Lyapunov-Krasovskii functionals (LKFs) \cite{XuL.IJSS08}. The key idea of these approaches is to construct suitable Lyapunov functionals that explicitly capture delay-dependent system dynamics. Considerable efforts have been devoted to reducing conservatism and improving computational tractability through various techniques, including the Newton–Leibniz formula \cite{XuL.IJSS08}, bounding techniques \cite{MooPKL.IJC01}, descriptor system formulations \cite{LeeMKP.Au04}, discretized LKF methods \cite{Gu.IJC97}, slack-variable techniques \cite{XuL.TAC07}, and delay-partitioning approaches \cite{DuLSW.IET09}.

Despite the extensive literature on Lyapunov-based delay analysis, applying these results to controller synthesis remains challenging. In particular, most delay control design problems involve nonlinear couplings between controller gain matrices and auxiliary variables such as Lyapunov matrices and slack variables. Consequently, the resulting synthesis conditions are typically non-convex and appear in the form of bilinear matrix inequalities (BMIs). This difficulty arises even in state-feedback problems and becomes more pronounced for output-feedback control. Representative examples include distributed state-feedback control for time-varying input delay systems \cite{YueH.Au05}, output-feedback control of delayed linear systems \cite{DuLS.Au10}, and $\mathcal{H}_\infty$ gain-scheduling control of delayed LPV systems \cite{WuG.Au01}. To address the associated non-convexity, existing approaches either impose structural constraints on decision variables to obtain convex formulations \cite{OliG.Au04,Fri.B14}, or rely on iterative numerical techniques such as the D–K iteration \cite{ZhouDG.B96}. However, the former approach often introduces additional conservatism, while the latter typically incurs significant computational cost and only guarantees locally optimal solutions. For systems with input delays, predictor-based control methods have also been proposed \cite{Krs.CSM10,Krs.TAC10}. Although these methods can effectively compensate for constant input delays, their implementation often requires full state information and the solution of partial differential equations, which may limit their practical applicability.

An alternative and increasingly attractive framework for delay system analysis is based on integral quadratic constraints (IQCs) \cite{MegR.TAC97}. IQC theory provides a powerful tool for characterizing uncertainties and nonlinearities in feedback interconnections and has been widely used in robust control analysis. In recent years, several works have developed IQC characterizations for delay operators, including both constant and time-varying delays \cite{JunS.TAC02,KaoR.Au07,Kao.TAC12}. Compared with classical Lyapunov-based approaches, IQC-based analysis often offers greater flexibility in handling multiple delays, dynamic uncertainties, and nonlinear effects, and can lead to improved stability margins.

Despite these developments, most existing IQC-based results for delayed systems focus primarily on stability and performance analysis, while controller synthesis within the IQC framework has received relatively limited attention except for \cite{VeeS.IJRNC14,YuaW.TCNS17}. For example, \cite{PfiS.IJRNC15} develops an LPV analysis framework that combines parameter-dependent Lyapunov functions with IQC descriptions of uncertainty, resulting in computationally efficient conditions for robust performance analysis. The results demonstrate that dynamic IQCs together with parameter-dependent Lyapunov functions can significantly improve performance compared with traditional analysis techniques. Recent studies have begun to explore IQC-based control synthesis for delayed systems. In particular, \cite{YuaW.TCYB16} develops a dynamic IQC-based exact-memory control framework for uncertain linear systems with time-varying state delays, leading to convex $\mathcal{H}_\infty$ synthesis conditions. Similarly, \cite{YuaW.Au17} proposes an exact-memory output-feedback control scheme for linear systems with time-varying input delays, resulting in convex LMI-based synthesis conditions for both memoryless and exact-memory controllers. These works demonstrate that the combination of IQC analysis and dissipation theory can effectively overcome the non-convexity difficulties commonly encountered in delay control design.

Motivated by recent developments in IQC-based delay analysis and control, this paper investigates the $\mathcal{H}_\infty$ output-feedback control problem for linear parameter-varying (LPV) systems with time-varying input delays. Compared with the linear time-invariant case, the LPV setting introduces additional challenges due to the interaction between parameter variations and delay dynamics. To address these challenges, we develop a control framework that integrates parameter-dependent Lyapunov functions with dynamic IQC characterizations of delay operators. An exact-memory controller architecture is employed, consisting of an LPV  output-feedback controller augmented with an internal delay loop. Similar in spirit to gain-scheduling control \cite{Sha.CSM02}, the embedded delay loop replicates the plant input delay within the controller, allowing the delay operator to be explicitly incorporated into the closed-loop interconnection. This structure enables the use of dissipation-based IQC analysis for controller synthesis. Based on this framework, new delay-dependent synthesis conditions are derived for LPV systems with time-varying input delays. By combining parameter-dependent Lyapunov functions with dynamic IQC multipliers, the resulting conditions can be formulated as a set of linear matrix inequalities (LMIs). Unlike conventional approaches \cite{YuaW.TCYB16,YuaW.Au17}, the proposed synthesis conditions do not explicitly involve controller gain variables, avoiding the need to parameterize controller gains as functions of scheduling parameters. An explicit controller reconstruction formula is provided to recover the LPV controller from the solution of the synthesis conditions. Consequently, the controller design problem can be formulated as a convex optimization problem that can be efficiently solved using standard semidefinite programming techniques \cite{BoyGFB.B04}. The proposed framework significantly simplifies delay control design and provides a systematic approach for synthesizing controllers for LPV systems with input delays. It also accommodates cases where the input delay is not directly measurable through a memoryless control structure. Numerical examples are provided to demonstrate the effectiveness and advantages of the proposed method.

Overall, the proposed IQC-based framework offers a fundamentally different perspective for the analysis and synthesis of delayed LPV systems. By leveraging the flexibility of IQC representations together with parameter-dependent Lyapunov functions, the method provides improved computational tractability and reduced conservatism compared with traditional Lyapunov–Krasovskii-based techniques \cite{MooPKL.IJC01,XuLZ.Au06,DuLS.Au10}. To the best of our knowledge, this paper presents the first IQC-based convex synthesis framework for $\mathcal{H}_\infty$ output-feedback control of LPV systems with time-varying input delays, providing a computationally tractable alternative to conventional Lyapunov–Krasovskii functional-based delay control methods. The main contributions of this paper are summarized as follows.

\begin{itemize}
\item
This paper develops a new control synthesis framework for linear parameter-varying (LPV) systems subject to time-varying input delays using the integral quadratic constraint (IQC) methodology. By integrating parameter-dependent Lyapunov functions with dynamic IQC characterizations of delay operators, the proposed approach provides a systematic method for stability and $\mathcal{H}_\infty$ performance analysis and controller synthesis of delayed LPV systems. This significantly extends existing IQC-based delay control techniques, which have primarily been developed for linear time-invariant systems.
\item
A key feature of the proposed method is the derivation of delay-dependent $\mathcal{H}_\infty$ output-feedback synthesis conditions that do not explicitly involve controller gain matrices. This avoids the commonly encountered difficulty of parameterizing controller gains as functions of scheduling parameters in LPV control design. The resulting synthesis conditions can be formulated as a set of linear matrix inequalities (LMIs), leading to a convex optimization problem that can be efficiently solved using semidefinite programming. An explicit controller reconstruction formula is also provided to recover the LPV controller from the solution of the synthesis conditions.
\item
The proposed framework incorporates parameter-dependent Lyapunov functions to explicitly capture the dependence of system dynamics on scheduling parameters. This allows the synthesis conditions to exploit parameter variation information more effectively and leads to improved stabilizability and closed-loop performance compared with approaches based on parameter-independent Lyapunov functions.
\end{itemize}

{\bf Notation}.
$\mathbf{RL}_\infty$ denotes the set of proper rational functions with real coefficients and no poles on the imaginary axis, and
$\mathbf{RH}\infty$ denotes the subset of functions in $\mathbf{RL}_\infty$ that are analytic in the closed right-half complex plane.
$\mathbf{RL}_\infty^{m\times n}$
and $\mathbf{RH}\infty^{m\times n}$ denote the sets of $m\times n$ transfer matrices with entries whose entries belong to  $\mathbf{RL}_\infty$ and $\mathbf{RH}_\infty$, respectively.
For a transfer matrix $G(s)$, its para-Hermitian conjugate is defined as $G^{\thicksim}(s):=G(-\bar{s})^*$.
$\mathbf{S}^n$ denotes the set of $n\times n$ real symmetric matrices, and $\mathbf{S}_+^n$ denotes the subset of positive-definite matrices.
For a square matrix $X$, ${\rm He}{X} := X + X^T$.
In symmetric block matrices, the symbol $\star$ is used to denote entries that are determined by symmetry.

For a vector $x \in \mathbf{C}^n$, $|x| := (x^*x)^{1/2}$ denotes the Euclidean norm.
The space $L_{2+}^n$ consists of square-integrable signals $u:[0,\infty) \rightarrow \mathbf{R}^n$ with norm
$\|u\|_2 := \left(\int^\infty_0 u^T(t) u(t) dt \right)^{1/2} < \infty$.
Given $u \in L_{2+}^n$, $u_T$ the truncated signal $u_T$ is defined as $u_T(t)=u(t)$ for $t\leq T$ and $u_T(t)=0$ otherwise.
The extended space $L_{2e+}^n$ consists of signals $u$ such that $u_T\in L_{2+}^n$ for all $T\geq 0$.

The remainder of the paper is organized as follows.
Section \ref{Sec.Prob} briefly reviews the IQC framework and formulates the control problem for LPV systems with time-varying input delays, together with the proposed delay-embedded controller structure.
Section \ref{Sec.Main} presents the main results, including convex synthesis conditions for delay-dependent (exact-memory) LPV output-feedback control, as well as a discussion of the corresponding memoryless control formulation.
Numerical examples demonstrating the effectiveness and reduced conservatism of the proposed approach are provided in Section \ref{Sec.Ex}.
Finally, concluding remarks are given in Section \ref{Sec.Con}.

\section{Problem Formulation}
\label{Sec.Prob}

IQCs provide a powerful framework for modeling a variety of nonlinearities and have been successfully applied to robust stability analysis of many dynamical systems (see, e.g., \cite{KaoR.Au07,Kao.TAC12,MegR.TAC97,Sei.TAC14}). Libraries of IQC multipliers for time-varying delays are available in \cite{KaoR.Au07} (continuous-time) and \cite{Kao.TAC12} (discrete-time).
We first recall a basic definition of dynamic IQCs.

\begin{definition}[\cite{Sei.TAC14}]
\label{Def.tIQC}
Let $\Pi \in \mathbf{RL}_\infty^{(m_1+m_2) \times (m_1+m_2)}$
be a proper, rational function, called a \emph{multiplier}, such that $\Pi = \Psi^{\thicksim} W \Psi$ with $W \in
\mathbf{S}^{n_z}$ and $\Psi \in \mathbf{RH}_\infty^{n_z \times
(m_1+m_2)}$.
Two signals $v \in L_{2e+}^{m_1}$ and $w \in L_{2e+}^{m_2}$
satisfy the IQC defined by the multiplier $\Pi$, and $(\Psi, W)$
is a hard IQC factorization of $\Pi$ if
\begin{align}
\label{Def-hardIQC-ine}
\int^T_0 z^T(t) W z(t) dt \geq 0, \quad \forall T \geq 0,
\end{align}
where $z \in \mathbf{R}^{n_z}$ is the filtered output of
$\Psi$ driven by $(v,w)$ with zero initial conditions.
A bounded, causal operator $\mathcal{S}: L_{2e+}^{m_1} \rightarrow L_{2e+}^{m_2}$ satisfies the IQC defined by $\Pi$ if condition (\ref{Def-hardIQC-ine}) holds for all $v \in L_{2e+}^{m_1}$, $w = \mathcal{S}(v)$ and all $T \geq 0$.
\end{definition}

We consider a class of linear parameter-varying (LPV) systems with bounded time-varying input delays:
\begin{align}
\label{Plant}
\mathcal{P}: \quad
    \begin{bmatrix}
        \dot{x}_p \\
        e \\
        y
    \end{bmatrix} &= \begin{bmatrix}
        A_p(\rho) & B_{p1}(\rho) & B_{p2}(\rho) \\
        C_{p1}(\rho) & D_{p11}(\rho) & D_{p12}(\rho) \\
        C_{p2}(\rho) & D_{p21}(\rho) & 0
    \end{bmatrix} \begin{bmatrix}
        x_p \\
        d \\
        \mathcal{D}_{\bar{\tau},r}(u)
    \end{bmatrix},
\end{align}
where $x_p \in \mathbf{R}^{n_p}$ is the plant state, $u \in \mathbf{R}^{n_u}$ is the control input, $e \in \mathbf{R}^{n_e}$ is the controlled output, $d \in \mathbf{R}^{n_d}$ is a disturbance, $y \in \mathbf{R}^{n_y}$ is the measurement output, and $\mathcal{D}_{\bar{\tau},r}(u) =  u(t - \tau(t))$ represents a time-varying delay $\tau(t)
\in {\cal T}_{\bar{\tau},r} := \{ \tau: \mathbf{R}_+ \rightarrow [0,\bar{\tau}], \ |\dot{\tau}(t)|\leq r \}$.
All state-space matrices are continuous functions of the scheduling parameter
$\rho \in {\cal P} \subset {\bf R}^s$, with parameter variation rate bounded by the convex polytope
\[
{\cal V} = \left\{ \nu: \underline{\nu}_k \leq \dot{\rho}_k \leq \bar{\nu}_k, k \in {\bf I}[1,s] \right\}.
\]

Define $w = \mathcal{S}_{\bar{\tau},r}(u) := u - \mathcal{D}_{\bar{\tau},r}(u)$, the LPV system (\ref{Plant}) can be equivalently expressed as a feedback interconnection:
\begin{align}
\label{Plant2}
\begin{aligned}
    \begin{bmatrix}
        \dot{x}_p \\
        e\\
        y
    \end{bmatrix}&=
    \begin{bmatrix}
        A_p(\rho) & -B_{p2}(\rho) & B_{p1}(\rho) & B_{p2}(\rho) \\
        C_{p1}(\rho) & -D_{p12}(\rho) & D_{p11}(\rho) & D_{p12}(\rho) \\
        C_{p2}(\rho) & 0 & D_{p21}(\rho) & 0
    \end{bmatrix}
    \begin{bmatrix}
        x_p\\
        w \\
        d\\
        u
    \end{bmatrix}, \\
    w &= \mathcal{S}_{\bar{\tau},r}(u).
\end{aligned}
\end{align}

We also make the following assumptions:
\begin{assumption}
\label{Ass1}
The triple $(A_p(\rho),B_{p2}(\rho),C_{p2}(\rho))$ is parametrically stabilizable and detectable. Additionally, $D_{p12}(\rho), D_{p21}(\rho)$ has full column and row rank, respectively.
\end{assumption}

\begin{assumption}
\label{Ass234}
The delay operator $\mathcal{S}_{\bar{\tau},r}$ satisfies a collection of IQCs defined by $\{\Pi_k\}_{k=1}^{N_\lambda} \subset \mathbf{RL}_\infty^{2n_u \times 2n_u}$, each admitting a $J_{n_u,n_u}$-spectral factorization $(\Psi_k,W_k)$ as detailed in \cite{Sei.TAC14}.
Specifically, we consider $(\Psi_k, W_k)$ in the form of $\Psi_k = \begin{bmatrix} \Psi_{11,k} & \Psi_{12,k} \\ 0 & I_{n_u} \end{bmatrix}$ and $W_k = \begin{bmatrix} X_k & 0 \\ 0 & -X_k \end{bmatrix}$ with $X_k \in \mathbf{S}_+^{n_u \times n_u}$.
\end{assumption}

Assumption \ref{Ass1} guarantees the existence of a stabilizing feedback controller for the non-delayed plant. Assumption \ref{Ass234} bounds the input-output behavior of the delay operator. Strict definiteness of the IQC blocks ensures the existence of a stable, minimum-phase spectral factorization \cite{Sei.TAC14}. These assumptions are standard and do not limit generality.

Assuming Assumption \ref{Ass234} holds, each IQC multiplier $\{\Pi_k\}_{k=1}^{N_\lambda}$ associated with the delay operator $\mathcal{S}_{\bar{\tau},r}$ can be factorized using a $J{n_u,n_u}$-spectral decomposition into a pair $(\Psi_k, W_k)$, following standard procedures \cite{Sei.TAC14, PfiS.Au15}.
Its corresponding state-space realization is given by the LTI system:
\begin{align}
\label{SF-Psi-LTI}
\begin{aligned}
    \begin{bmatrix}
        \dot{x}_{\psi} \\
        z_{1,k} \\
        z_{2,k}
    \end{bmatrix} &= \begin{bmatrix}
        A_{\psi} & B_{\psi1} & B_{\psi2} \\
        C_{\psi,k} & D_{\psi1,k} & D_{\psi2,k} \\
        0 & 0 & I
    \end{bmatrix} \begin{bmatrix}
        x_\psi \\
        u \\
        w
    \end{bmatrix}, \quad k \in {\bf I}[1,N_\lambda]
\end{aligned}
\end{align}
where $x_\psi \in \mathbf{R}^{n_{\psi}}$ is the state vector of $\Psi_k$ with $x_\psi(0)=0$, and $z_{1,k}, z_{2,k} \in \mathbf{R}^{n_u}$ are the outputs.

For controller implementation, we assume that the time-varying delay $\tau(t)$ is measurable in real time, which enables the proposed exact-memory controller structure. We also briefly consider the case where delay information is unavailable, which leads to a memoryless controller; however, as discussed later, the associated synthesis problem is non-convex and more challenging.

Our objective is to design an output-feedback controller for (\ref{Plant}) such that
the resulting closed-loop system is stable for all admissible $(\rho, \dot{\rho}) \in \mathcal{P} \times \mathcal{V}$ and input delays $\tau(t) \in {\cal T}_{\bar{\tau},r}$, and
the closed-loop system achieves a desired $\mathcal{L}_2$-gain performance from $d$ to $e$.
The associated synthesis conditions are formulated under the IQC framework as convex LMIs on all design variables, including the scaling matrices of the IQC multipliers $\{\Pi_k\}_{k=1}^{N_\lambda}$.

To achieve this, we employ a delay-dependent controller structure following \cite{YuaW.Au17}, which integrates a standard LPV output-feedback law with an internal delay loop. As illustrated in Fig. \ref{Fig.CL}, the loop explicitly incorporates the delayed input signal, enabling online scheduling of the controller gains based on the real-time delay $\tau(t)$.

\begin{figure}[htb]
\centering
\includegraphics[width=0.55\textwidth]{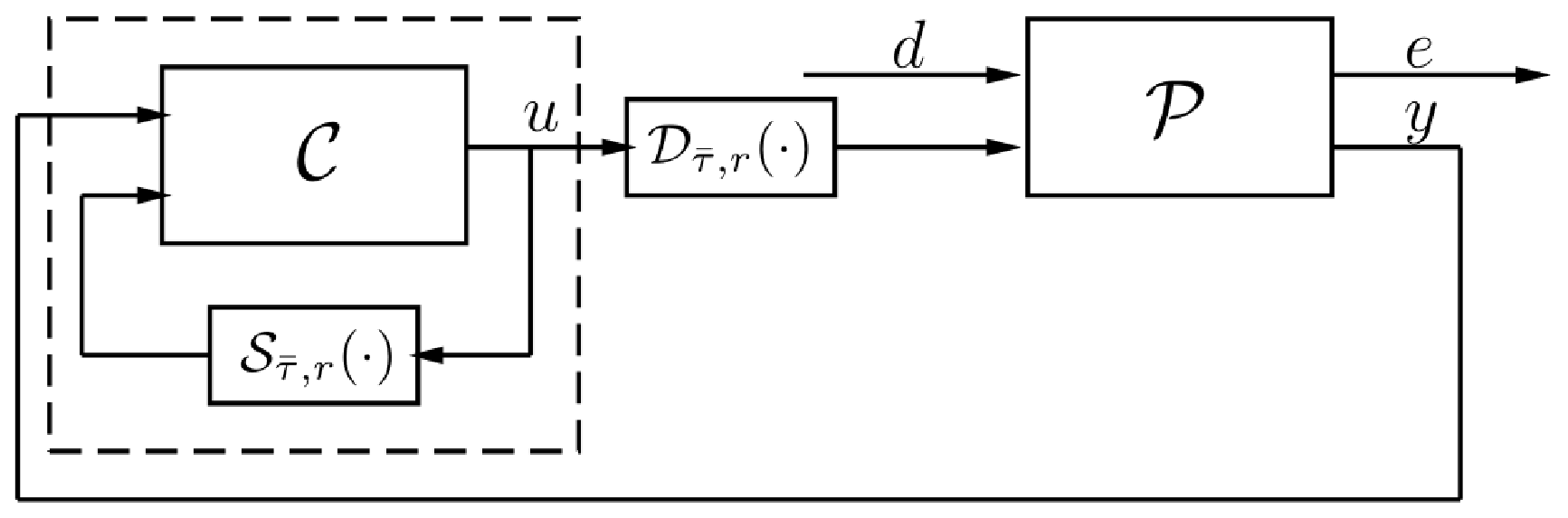}
\caption{The feedback system with a delay loop.}
\label{Fig.CL}
\end{figure}

In the block diagram, the $\mathcal{P}$ block represents the LPV plant with a time-varying input delay $\mathcal{D}_{\bar{\tau},r}(\cdot)$. The proposed control strategy (dashed box in Fig. \ref{Fig.CL}) combines a linear feedback law $\mathcal{C}$ from $y$ to $u$ with the internal delay loop via the nonlinear delay operator $\mathcal{S}_{\bar{\tau},r}(\cdot)$. This embodies the notion of exact-memory control, where the loop effectively “memorizes” all past input signals over $[t-\bar{\tau}, t]$, allowing the controller to adjust its gains in real time. Such a structure is critical for enabling convex synthesis conditions.

Controllers with embedded delay loops of this type are widely adopted in the literature (see, e.g., \cite{OliG.Au04, Fri.B14, Bri.B15, WuG.Au01}) and are practical in many engineering applications, including chemical processes, robotic systems, and internal combustion engines, where delays depend on measurable operating conditions or system parameters. By leveraging this exact-memory structure, the proposed framework provides a systematic, tractable, and performance-guaranteed approach for output-feedback control of LPV systems with time-varying input delays.

%

\section{Output Feedback Control Synthesis}
\label{Sec.Main}

In this section, we study the LPV delay control problem under the IQC framework using two control strategies and highlight the advantages of the proposed delay-dependent control structure. We first present the exact-memory control synthesis.

\subsection{Exact-Memory Control}

By combining the plant dynamics (\ref{Plant}) with the IQC-induced dynamics (\ref{SF-Psi-LTI}), we define the augmented state $x_{aug} = [x_p^T \ x_\psi^T]^T$ of dimension $n_{aug} = n_p + n_{\psi}$. The augmented system can be compactly written as:
\begin{align*}
\begin{bmatrix}
\dot{x}_{aug} \\ z_{1,k} \\ z_{2,k} \\ e \\ y \end{bmatrix} &= \left[ \begin{array}{cc:ccc}
    A_p(\rho) & 0 & -B_{p2}(\rho) & B_{p1}(\rho) & B_{p2}(\rho) \\
    0 & A_\psi & B_{\psi2} & 0 & B_{\psi1} \\ \hdashline
    0 & C_{\psi,k} & D_{\psi2,k} & 0 & D_{\psi1,k} \\
    0 & 0 & I & 0 & 0 \\
    C_{p1}(\rho) & 0 & -D_{p12}(\rho) & D_{p11}(\rho) & D_{p12}(\rho) \\
    C_{p2}(\rho) & 0 & 0 & D_{p21}(\rho) & 0
\end{array} \right]
\begin{bmatrix}
x_{aug} \\ w \\ d \\ u
\end{bmatrix} \\
&:= \begin{bmatrix}
    A_{aug}(\rho) & B_{aug0}(\rho) & B_{aug1}(\rho) & B_{aug2}(\rho) \\
    C_{aug0,k} & D_{\psi2,k} & 0 & D_{\psi1,k} \\
    0 & I & 0 & 0 \\
    C_{aug1}(\rho) & -D_{p12}(\rho) & D_{p11}(\rho) & D_{p12}(\rho) \\
    C_{aug2}(\rho) & 0 & D_{p21}(\rho) & 0
\end{bmatrix} \begin{bmatrix}
    x_{aug} \\ w \\ d \\ u
\end{bmatrix}.
\end{align*}

The exact-memory LPV output-feedback controller (Fig. \ref{Fig.CL}) is defined as:
\begin{align}
\label{OF-Controller}
\mathcal{C}: \quad \begin{aligned}
    \begin{bmatrix}
        \dot{x}_c \\
        u
    \end{bmatrix} &= \begin{bmatrix}
        A_c(\rho) & B_{c1}(\rho) & B_{c2}(\rho) \\
        C_c(\rho) & D_{c1}(\rho) & D_{c2}(\rho)
    \end{bmatrix} \begin{bmatrix}
        x_c \\
        y \\
        \mathcal{S}_{\bar{\tau},r}(u)
    \end{bmatrix},
\end{aligned}
\end{align}
where $x_c \in \mathbf{R}^{n_c}$ is the controller state, and the controller gain matrices $A_c (\cdot), B_{c1}(\cdot), B_{c2}(\cdot), C_c(\cdot)$, $D_{c1}(\cdot), D_{c2}(\cdot)$ are design variables.

The extended closed-loop system absorbing IQC dynamics is written as:
\begin{align}
\label{SF-CloseLoop}
\begin{aligned}
    \begin{bmatrix}
        \dot{x}_{cl} \\
        z_{1,k} \\
        z_{2,k} \\
        e
    \end{bmatrix} &= \begin{bmatrix}
        A_{cl}(\rho) & B_{cl1}(\rho) & B_{cl2}(\rho) \\
        C_{cl1,k}(\rho) & D_{cl11,k}(\rho) & D_{cl12,k}(\rho) \\
        0 & I & 0 \\
        C_{cl2}(\rho) & D_{cl21}(\rho) & D_{cl22}(\rho)
    \end{bmatrix} \begin{bmatrix}
        x_{cl} \\
        w \\
        d
    \end{bmatrix}, \\
    w &= \mathcal{S}_{\bar{\tau},r}(u)
\end{aligned}
\end{align}
with $x_{cl} = [x_{aug}^T \ x_c^T]^T \in \mathbf{R}^{n_{aug}+n_c}$. The system matrices are defined in the following equation
\begin{align}
\label{OF-CloseLoop}
\begin{bmatrix}
   A_{cl}(\rho) & B_{cl1}(\rho) & B_{cl2}(\rho) \\
   C_{cl1,k}(\rho) & D_{cl11,k}(\rho) & D_{cl12,k}(\rho) \\
   C_{cl2}(\rho) & D_{cl21}(\rho) & D_{cl22}(\rho)
\end{bmatrix} &= \left[ \begin{array}{cc:cc}
   A_{aug}(\rho) & 0 & B_{aug0}(\rho) & B_{aug1}(\rho) \\
   0 & 0 & 0 & 0 \\ \hdashline
   C_{aug0,k} & 0 & D_{\psi2,k} & 0 \\
   C_{aug1}(\rho) & 0 & -D_{p12}(\rho) & D_{p11}(\rho)
\end{array} \right] \nonumber \\
& \hspace*{-1.25in} + \left[ \begin{array}{cc}
   0 & B_{aug2}(\rho) \\
   I & 0 \\ \hdashline
   0 & D_{\psi1,k} \\
   0 & D_{p12}(\rho)
\end{array} \right] \begin{bmatrix}
    A_c(\rho) & B_{c2}(\rho) & B_{c1}(\rho) \\
    C_{c}(\rho) & D_{c2}(\rho) & D_{c1}(\rho)
\end{bmatrix} \left[ \begin{array}{cc:cc}
    0 & I & 0 & 0 \\
    0 & 0 & I & 0 \\
    C_{aug2}(\rho) & 0 & 0 & D_{p21}(\rho)
\end{array} \right].
\end{align}

With this formulation, the time-domain IQC and dissipation inequality framework \cite{Sei.TAC14} can be directly applied to derive convex synthesis conditions for the exact-memory output-feedback controller.
The following theorem provides the synthesis conditions for $\mathcal{H}_\infty$ exact-memory LPV output-feedback control of the input-delay LPV system under the IQC and dissipation inequality framework.

\begin{theorem}
\label{OF-Thm-Syn}
Consider the input-delayed LPV plant (\ref{Plant}). If there exist positive-definite matrix functions $R, S: \mathbf{R}_+ \rightarrow \mathbf{S}_+^{n_{aug}}$, positive-definite matrix functions $\hat{X}_k \in \mathbf{S}_+^{n_u}, k \in {\bf I}[1, N_{\lambda}]$, a rectangular matrix function $X \in \mathbf{R}^{n_u \times n_u}$, and a positive scalar $\gamma$ such that
\begin{align}
{\cal N}_R^T \begin{bmatrix}
    \left\{ \begin{matrix} {\rm He} \{ A_{aug}(\rho) R(\rho) \} \\ -\{ \underline{\nu}, \bar{\nu} \} \frac{\partial R}{\partial \rho} \end{matrix} \right\} & \star & \star & \star & \star \\
    X^T(\rho) B_{aug0}^T(\rho) & -\displaystyle{\sum^{N_\lambda}_{k=1}} (X(\rho) + X^T(\rho) - X_k(\rho)) & \star & \star & \star \\
    B_{aug1}^T(\rho) & 0 & -\gamma I_{n_d} & \star & \star \\
    C_{aug0} R(\rho) & D_{\psi2} X(\rho) & 0 & -\Lambda(\rho) & \star \\
    C_{aug1}(\rho) R(\rho) & -D_{p12}(\rho) X(\rho) & D_{p11}(\rho) & 0 & -\gamma I_{n_e}
    \end{bmatrix} {\cal N}_R &< 0 \label{OF-Thm-Syn-LMI1} \\
{\cal N}_S^T \begin{bmatrix}
    {\rm He} \{ S(\rho) A_{aug}(\rho) \} + \{ \underline{\nu}, \bar{\nu} \} \frac{\partial S}{\partial \rho} & \star & \star & \star \\
    C_{aug0} & -\Lambda(\rho) & \star & \star \\
    C_{aug1}(\rho) & 0 & -\gamma I_{n_e} & \star \\
    B_{aug1}^T(\rho) S(\rho) & 0 & D_{p11}^T(\rho) & -\gamma I_{n_d}
    \end{bmatrix} {\cal N}_S &< 0 \label{OF-Thm-Syn-LMI2} \\
\begin{bmatrix}
    R(\rho) & I_{n_{aug}} \\
    I_{n_{aug}} & S(\rho)
    \end{bmatrix} &> 0 \label{OF-Thm-Syn-LMI3}
\end{align}
where ${\cal N}_R = {\rm Ker} \begin{bmatrix} B_{aug2}^T(\rho) & 0 & 0 & D_{\psi1}^T & D_{p12}^T(\rho) \end{bmatrix}, {\cal N}_S = {\rm Ker} \begin{bmatrix} C_{aug2}(\rho) & 0 & 0 & D_{p21}(\rho) \end{bmatrix}$, and $\Lambda = {\rm diag} \left\{ X_1(\rho), \cdots, X_k(\rho) \right\}$.
Then the closed-loop system (\ref{SF-CloseLoop}) is asymptotically stable with its $\mathcal{L}_2$ gain less than $\gamma$.
\end{theorem}

\begin{proof}
We consider the extended closed-loop system in the standard feedback interconnection form (\ref{SF-CloseLoop}) together with the LPV controller structure (\ref{OF-CloseLoop}).

According to \cite{PfiS.IJRNC15}, the closed-loop system with the delay loop is internally stable and achieves an $\mathcal{L}_2$-gain less than $\gamma$ if there exist a matrix function $P \in \mathbf{S}_+^{n_{aug}+n_c}$, and positive definite matrix functions $X_k, k\in {\bf I}[1,N_\lambda]$   such that
\begin{align}
\label{SF-Prf-Syn-LMI*}
\begin{aligned}
    &\begin{bmatrix}
        {\rm He} \{PA_{cl}\}+\dot{P} & \star & \star \\
        B_{cl1}^T P & 0 & \star  \\
        B_{cl2}^T P & 0 & -\gamma I_{n_d}
    \end{bmatrix} + \frac{1}{\gamma}
    \begin{bmatrix}
        C_{cl2}^T \\
        D_{cl21}^T \\
        D_{cl22}^T
    \end{bmatrix} \begin{bmatrix}
        C_{cl2} & D_{cl21} & D_{cl22}
    \end{bmatrix} \\
    &\hspace*{1.25in} + \sum^{N_\lambda}_{k=1} \begin{bmatrix}
        C_{cl1,k}^T & 0 \\
        D_{cl11,k}^T & I \\
        D_{cl12,k}^T & 0
    \end{bmatrix} \begin{bmatrix} X_k^{-1} & \\ & -X_k^{-1} \end{bmatrix} \begin{bmatrix}
        C_{cl1,k} & D_{cl11,k} & D_{cl12,k} \\
        0 & I & 0
    \end{bmatrix} < 0.
\end{aligned}
\end{align}
Partition the Lyapunov matrix $P$ and its inverse as
\begin{align*}
P = \begin{bmatrix}
S & N \\
N^T & \star
\end{bmatrix}, \qquad P^{-1} = \begin{bmatrix}
R & M \\
M^T & \star
\end{bmatrix}
\end{align*}
where $R, S \in \mathbf{S}_+^{n_{aug}}$.
It follows that $P>0$ if and only if
$\begin{bmatrix} R & I \\ I & S \end{bmatrix} > 0$, which yields condition (\ref{OF-Thm-Syn-LMI3}).

Substituting the closed-loop matrices (\ref{OF-CloseLoop}) together with the IQC scaling matrices $\{X_k\}_{k=1}^{N_\lambda}$ satisfying Assumption \ref{Ass234} into (\ref{SF-Prf-Syn-LMI*}), and applying the Schur complement, we obtain the equivalent condition
\begin{align}
\begin{bmatrix}
    {\rm He} \{PA_{cl}\}+\dot{P} & \star & \star & \star & \star \\
    B_{cl1}^T P & -\sum^{N_\lambda}_{k=1} X_k^{-1} & \star & \star & \star \\
    B_{cl2}^T P & 0 & -\gamma I_{n_d} & \star & \star \\
    C_{cl1} & D_{cl11} & D_{cl12} & -\Lambda & \star \\
    C_{cl2} & D_{cl21} & D_{cl22} & 0 & -\gamma I_{n_e}
\end{bmatrix} < 0, \label{SF-Prf-Syn-LMI**}
\end{align}
where $\Lambda = {\rm diag} \left\{ X_1, \cdots, X_k \right\}$.
The matrices $C_{cl1}$, $D_{cl11}$, and $D_{cl12}$ are formed by stacking the corresponding subsystem matrices.

Using the inequality
$-X_k^{-1} \leq - X^{-1} (X + X^T - X_k) X^{-T}$ for any nonsingular matrix $X$, condition (\ref{SF-Prf-Syn-LMI**}) is implied by the sufficient condition
\begin{align}
\begin{bmatrix}
    {\rm He} \{PA_{cl}\} + \dot{P} & \star & \star & \star & \star \\
    X^T B_{cl1}^T P & -\sum^{N_\lambda}_{k=1} (X+X^T -X_k) & \star & \star & \star \\
    B_{cl2}^T P & 0 & -\gamma I_{n_d} &\star & \star \\
    C_{cl1} & D_{cl11} X & D_{cl12} &  -\Lambda & \star \\
    C_{cl2} & D_{cl21} X & D_{cl22} & 0 & -\gamma I_{n_e}
\end{bmatrix} < 0 \label{SF-Prf-Syn-LMI3}.
\end{align}

The left-hand side of condition (\ref{SF-Prf-Syn-LMI3}) can be expressed as
\begin{align*}
\Gamma + U \Sigma V^T + V \Sigma^T U^T < 0,
\end{align*}
where $\Gamma$, and $U$, $V$ are defined as
\begin{align*}
\Gamma &= \begin{bmatrix}
    P \begin{bmatrix} A_{aug} & 0 \\ 0 & 0 \end{bmatrix} + \begin{bmatrix} A_{aug}^T & 0 \\ 0 & 0 \end{bmatrix} P + \dot{P} & \star & \star & \star & \star \\
    X \begin{bmatrix} B_{aug0}^T & 0 \end{bmatrix} P & -\sum^{N_\lambda}_{k=1} (X + X^T - X_k) & \star & \star & \star \\
    \begin{bmatrix} B_{aug1}^T & 0 \end{bmatrix} P & 0 & -\gamma I & \star & \star \\
    \begin{bmatrix} C_{aug0} & 0 \end{bmatrix} & D_{\psi2} X & 0 & -\Lambda & \star \\
    \begin{bmatrix} C_{aug1} & 0 \end{bmatrix} & -D_{p12} X & D_{p11} & 0 & -\gamma I
\end{bmatrix}, \\
U &= \begin{bmatrix}
    P \begin{bmatrix} 0 & B_{aug2} \\ I & 0 \end{bmatrix} \\
    0 \\
    0 \\
    \begin{bmatrix} 0 & D_{\psi1} \end{bmatrix} \\
    \begin{bmatrix} 0 & D_{p12} \end{bmatrix}
\end{bmatrix} \qquad
V^T = \begin{bmatrix}
    \begin{bmatrix} 0 & I \\ 0 & 0 \\ C_{aug2} & 0 \end{bmatrix} & \begin{bmatrix} 0 \\ I \\ 0 \end{bmatrix} & \begin{bmatrix} 0 \\ 0 \\ D_{p21} \end{bmatrix} & 0 & 0
\end{bmatrix}
\end{align*}
with $\Sigma = \begin{bmatrix}
    A_c & B_{c2} & B_{c1} \\
    C_c & D_{c2} & D_{c1}
\end{bmatrix}$.
The matrices $U$ and $V$ contain the plant-dependent terms, while $\Sigma$ collects all controller matrices.
The kernel spaces of $U$ and $V$ are given by
\begin{align*}
U_{\perp} &= {\rm diag} \{ P^{-1}, I, I, I, I \} \left[ \begin{array}{ccc}
    W_1 & 0 & 0 \\
    0 & 0 & 0 \\ \hdashline
    0 & I & 0 \\
    0 & 0 & I \\
    W_2 & 0 & 0 \\
    W_3 & 0 & 0
\end{array} \right], \qquad \begin{bmatrix} W_1 \\ W_2 \\ W_3 \end{bmatrix} = {\rm Ker} \begin{bmatrix} B_{aug2} \\ D_{\psi1} \\ D_{p12} \end{bmatrix}
\end{align*}
\begin{align*}
V_{\perp} &= \left[ \begin{array}{ccc}
    Y_1 & 0 & 0 \\
    0 & 0 & 0 \\ \hdashline
    0 & 0 & 0 \\
    Y_2 & 0 & 0 \\
    0 & I & 0 \\
    0 & 0 & I
\end{array} \right], \qquad \begin{bmatrix} Y_1 \\ Y_2 \end{bmatrix} = {\rm Ker} \begin{bmatrix} C_{aug2}^T \\ D_{p21}^T \end{bmatrix}.
\end{align*}
Applying the elimination lemma yields the equivalent conditions \begin{align}
U_{\perp}^T \Gamma U_{\perp} < 0 \quad \mbox{and} \quad V_{\perp}^T \Gamma V_{\perp} < 0.
\end{align}
It can be verified directly that the first condition corresponds to eq. (\ref{OF-Thm-Syn-LMI1}), while the second condition corresponds to eq. (\ref{OF-Thm-Syn-LMI2}). This completes the proof.
\end{proof}

It is clear from Theorem \ref{OF-Thm-Syn} that the proposed exact-memory output-feedback controller structure (\ref{OF-Controller}) leads to convex synthesis conditions under the IQC framework. Specifically, given a finite collection of multipliers $\{\Pi_k\}_{k=1}^{N_\lambda}$ and their corresponding state-space realizations in the form of (\ref{SF-Psi-LTI}), the synthesis conditions (\ref{OF-Thm-Syn-LMI1})-(\ref{OF-Thm-Syn-LMI3}) can be expressed as parameter-dependent LMIs in terms of the Lyapunov matrices and the scaling matrices $\{X_k\}_{k=1}^{N_\lambda}$. Consequently, the controller synthesis problem can be formulated as a convex optimization problem.
It is worth noting that the above synthesis conditions provide an alternative formulation compared with the approach in \cite{YuaW.Au17}. In particular, the LMIs derived in Theorem \ref{OF-Thm-Syn} do not explicitly depend on the controller gain matrices. This feature eliminates the need to impose a prescribed functional parameterization on the controller gains with respect to the scheduling variable $\rho$. As a result, the controller gains can be reconstructed afterward without restricting their parameter dependence a priori. This property is particularly attractive in LPV control design, since enforcing specific functional structures on controller gains may introduce unnecessary conservatism or lead to additional computational complexity.

Since the matrices $R(\cdot)$, $S(\cdot)$, $X(\cdot)$, and $X_k(\cdot)$ are parameter-dependent functions of the scheduling variable $\rho$, they must be approximated using suitable basis-function expansions. In particular, we consider the following parameterizations:
\begin{align*}
R(\rho) &= \sum_{i=1}^{n_f} f_i(\rho) R_i, \qquad S(\rho) = \sum_{i=1}^{n_g} g_i(\rho) S_i \\
X(\rho) &= \sum_{i=1}^{n_h} h_i(\rho) X_i, \qquad X_k(\rho) = \sum_{i=1}^{n_h} h_i(\rho) X_{k,i}
\end{align*}
where ${f_i(\rho)}$, ${g_i(\rho)}$, and ${h_i(\rho)}$ denote selected basis functions over the parameter set ${\cal P}$, and $R_i$, $S_i$, $X_i$, and $X_{k,i}$ are constant matrices to be determined.

To enforce the parameter-dependent LMIs over the entire parameter domain ${\cal P}$, a finite set of grid points is selected within ${\cal P}$ and the LMIs are evaluated at each grid point. This procedure converts the LPV synthesis conditions into a finite set of LMIs, leading to the following convex semidefinite programming problem:
\begin{equation}
\label{Opt}
\begin{aligned}
    &\min_{R_i, S_i, X_i, X_{k,i}, \ \forall k \in {\bf I}[1,N_\lambda]} \qquad \gamma \\
    &\hspace*{0.5in} \mbox{s.t.} \qquad \mbox{conditions (\ref{OF-Thm-Syn-LMI1})-(\ref{OF-Thm-Syn-LMI3}).}
\end{aligned}
\end{equation}
Solving this optimization problem yields the matrices defining the parameter-dependent functions $R(\rho)$, $S(\rho)$, $X(\rho)$, and $X_k(\rho)$. Based on these solutions, the LPV output-feedback controller can then be explicitly constructed using the procedure provided in the following theorem.

\begin{theorem}
\label{thm.OF-Thm-Syn-ABCD}
Suppose the parameter-dependent LMIs in Theorem 1 admit a feasible solution for a given performance level $\gamma$. Then the controller gains of a full-order ($n_c = n_{aug}$) LPV output-feedback controller of the form (\ref{OF-Controller}) can be explicitly constructed through the following procedure.
\begin{description}
\item[Step 1]
Determine parameter-dependent matrices
$\hat{D}_{c2}(\rho), \hat{D}_{c1}(\rho)$ that satisfy the feasibility condition associated with the LMI constraints obtained in Theorem 1.
{\small
\begin{align}
\label{hatD}
\Pi = -\begin{bmatrix}
-\displaystyle{\sum^{N_\lambda}_{k=1}} (X(\rho) + X^T(\rho) - X_k(\rho)) & \star & \star & \star \\
0 & -\gamma I & \star & \star \\
D_{\psi2} X(\rho) + D_{\psi1} \hat{D}_{c2}(\rho) & D_{\psi1} \hat{D}_{c1}(\rho) D_{p21}(\rho) & -\Lambda(\rho) & \star \\
-D_{p12}(\rho) X(\rho) + D_{p12}(\rho) \hat{D}_{c2}(\rho) & D_{p11}(\rho) + D_{p12}(\rho) \hat{D}_{c1}(\rho) D_{p21}(\rho) & 0 & -\gamma I
\end{bmatrix} > 0.
\end{align}
}
\item[Step 2]
Using the matrices obtained from the LMI solution together with the feedthrough matrices from Step 1, compute the intermediate matrices
$\hat{B}_{c1}(\rho), \hat{B}_{c2}(\rho)$ and $\hat{C}_c(\rho)$ from a set of linear matrix equations
{\small
\begin{align}
\left[ \begin{array}{c:c}
    0 & \begin{matrix} I & 0 & 0 & 0 \\ 0 & D_{p21}(\rho) & 0 & 0 \end{matrix} \\ \hdashline
    \begin{matrix} I & 0 \\ 0 & D_{p21}^T(\rho) \\ 0 & 0 \\ 0 & 0 \end{matrix} & -\Pi
\end{array} \right] \left[ \begin{array}{c}
    \hat{B}_{c2}^T(\rho) \\
    \hat{B}_{c1}^T(\rho) \\ \hdashline
    \star
\end{array} \right] &= - \left[ \begin{array}{c}
    0 \\
    C_{aug2}(\rho) \\ \hdashline
    0 \\
    B_{aug1}^T(\rho) S(\rho) \\
    C_{aug0} + D_{\psi1} \hat{D}_{c1}(\rho) C_{aug2}(\rho) \\
    C_{aug1}(\rho) + D_{p12}(\rho) \hat{D}_{c1}(\rho) C_{aug2}(\rho)
\end{array} \right] \label{hatB} \\
\left[ \begin{array}{c:c}
    0 & \begin{matrix} 0 & 0 & D_{\psi1}^T & D_{p12}^T(\rho) \end{matrix} \\ \hdashline
    \begin{matrix} 0 \\ 0 \\ D_{\psi1} \\ D_{p12}(\rho) \end{matrix} & -\Pi
\end{array} \right] \left[ \begin{array}{c}
    \hat{C}_{c}(\rho) \\ \hdashline
    \star
\end{array} \right] &= - \left[ \begin{array}{c}
    B_{aug2}^T(\rho) \\ \hdashline
    X^T(\rho) B_{aug0}^T(\rho) + \hat{D}_{c2}^T(\rho) B_{aug2}^T(\rho) \\
    B_{aug1}^T(\rho) + D_{p21}^T(\rho) \hat{D}_{c1}^T(\rho) B_{aug2}^T(\rho) \\
    C_{aug0} R(\rho) \\
    C_{aug1}(\rho) R(\rho) \\
\end{array} \right]. \label{hatC}
\end{align}
}
These equations can be solved using standard least-squares methods.
\item[Step 3]
The intermediate controller state matrix  $\hat{A}_c(\rho,\dot{\rho})$ is then obtained from the closed-loop consistency relations
\begin{align}
\label{hatA}
\hat{A}_c(\rho,\dot{\rho}) &= -A_{aug}^T(\rho) - C_{aug2}^T(\rho) \hat{D}_{c1}^T(\rho) B_{aug2}^T(\rho) + S(\rho) \frac{d R}{d t} + N(\rho) \frac{d M^T}{d t} + L_2 \Pi^{-1} L_1^T
\end{align}
with
\begin{align*}
\begin{bmatrix} L_1^T & L_2^T \end{bmatrix} & = \begin{bmatrix} X^T(\rho) B_{aug0}^T(\rho) + \hat{D}_{c2}^T(\rho) B_{aug2}^T(\rho) & \hat{B}_{c2}^T(\rho) \\
B_{aug1}^T(\rho) + D_{p21}^T(\rho) \hat{D}_{c1}^T(\rho) B_{aug2}^T(\rho) & B_{aug1}^T(\rho) S(\rho) + D_{p21}^T(\rho) \hat{B}_{c1}^T(\rho) \\
C_{aug0} R(\rho) + D_{\psi1} \hat{C}_c(\rho) & C_{aug0} + D_{\psi1} \hat{D}_{c1}(\rho) C_{aug2}(\rho) \\
C_{aug1}(\rho) R(\rho) + D_{p12}(\rho) \hat{C}_c(\rho) &  C_{aug1}(\rho) + D_{p12}(\rho) \hat{D}_{c1}(\rho) C_{aug2}(\rho)
\end{bmatrix}.
\end{align*}
\item[Step 4]
Finally, Finally, the actual controller matrices are obtained through a coordinate transformation
{\small
\begin{align}
\begin{bmatrix}
A_c(\rho,\dot{\rho}) & B_{c2}(\rho) & B_{c1}(\rho) \\
C_c(\rho) & D_{c2}(\rho) & D_{c1}(\rho)
\end{bmatrix} &= \begin{bmatrix}
    N(\rho) & S(\rho) B_{aug2}(\rho) \\
    0 & I
\end{bmatrix}^{-1} \nonumber \\
& \hspace*{-1.5in} \times \begin{bmatrix}
    \left\{ \begin{matrix} \hat{A}_c(\rho,\dot{\rho}) \\ - S(\rho) A_{aug}(\rho) R(\rho) \end{matrix} \right\} & \left\{ \begin{matrix} \hat{B}_{c2}(\rho) \\ - S(\rho) B_{aug0}(\rho) X(\rho) \end{matrix} \right\} & \hat{B}_{c1}(\rho) \\
    \hat{C}_c(\rho) & \hat{D}_{c2}(\rho) & \hat{D}_{c1}(\rho)
\end{bmatrix} \begin{bmatrix}
M^T(\rho) & 0 & 0 \\
0 & X(\rho) & 0 \\
C_{aug2}(\rho) R(\rho) & 0 & I
\end{bmatrix}^{-1}, \label{OF-Thm-Syn-ABCD}
\end{align}
}
where $M(\rho) N^T(\rho) = I - R(\rho) S(\rho)$.
\end{description}
\end{theorem}

\begin{proof}
To convert the synthesis condition in Theorem \ref{OF-Thm-Syn}  into a convex LMI with respect to all design variables, we set the controller order as $n_c = n_{aug}$ and partition the Lyapunov matrix as
\begin{align*}
P &= \begin{bmatrix}
    S & N \\
    N^T & X^{-1}
\end{bmatrix},
\end{align*}
where $S, X^{-1} \in \mathbf{S}_+^{n_{aug}}$. Define the transformation matrices
\begin{align*}
Z_1 &= \begin{bmatrix}
    R & I \\
    M^T & 0
\end{bmatrix}, \quad Z_2 = \begin{bmatrix}
    I & S \\
    0 & N^T
\end{bmatrix}.
\end{align*}
These matrices satisfy $P Z_1 = Z_2$ and $M N^T = I_{n_{aug}} - R S$, which further implies $X^{-1} = -N^T R M^{-T}$.
From condition (\ref{OF-Thm-Syn-LMI3}) it follows that $Z_1^T P Z_1 = \begin{bmatrix}
    R & I \\
    I & S \end{bmatrix} > 0$,
Since $Z_1$ is nonsingular, this condition guarantees
$P > 0$.

Pre- and post-multiplying inequality (\ref{SF-Prf-Syn-LMI3}) by ${\rm diag} \{Z_1, I_{n_u}, I_{n_d}, I_{n_e},I_{N_\lambda n_u}\}$ and its transpose yields the following transformed terms:
\begin{align*}
Z_1^T \frac{d P}{d t} Z_1 &= \begin{bmatrix}
    -\frac{d R}{d t} & \star \\
    -\left( S \frac{d R}{d t} + N \frac{d M^T}{d t} \right) & \frac{d S}{d t}
\end{bmatrix} \\
Z_1^T P A_{cl}Z_1 &= Z_2^T A_{cl} Z_1 \\
&= \begin{bmatrix}
    A_{aug} R + B_{aug2} \hat{C}_c & A_{aug} + B_{aug2} \hat{D}_{c1} C_{aug2} \\
    \hat{A}_c & S A_{aug} + \hat{B}_{c1} C_{aug2}
\end{bmatrix}, \\
X^T B_{cl1}^T P Z_1 &= X^T B_{cl1}^T Z_2 \\
&= \begin{bmatrix}
    X^T B_{aug0}^T + \hat{D}_{c2}^T B_{aug2}^T & \hat{B}_{c2}^T
\end{bmatrix}, \\
B_{cl2}^T P Z_1 &= B_{cl2}^T Z_2 \\
&= \begin{bmatrix}
    B_{aug1}^T + D_{p21}^T \hat{D}_{c1}^T B_{aug2}^T & B_{aug1}^T S + D_{p21}^T \hat{B}_{c1}^T
\end{bmatrix}, \\
C_{cl1} Z_1 &= \begin{bmatrix}
    C_{aug0} R + D_{\psi1} \hat{C}_c & C_{aug0} + D_{\psi1} \hat{D}_{c1} C_{aug2}
\end{bmatrix} \\
C_{cl2} Z_1 &= \begin{bmatrix}
    C_{aug1} R + D_{p12} \hat{C}_c & C_{aug1} + D_{p12} \hat{D}_{c1} C_{aug2}
\end{bmatrix}, \\
D_{cl11} X &= D_{\psi2} X + D_{\psi1} \hat{D}_{c2}, \qquad D_{cl12} = D_{\psi1} \hat{D}_{c1} D_{p21} \\
D_{cl21} X &= -D_{p12} X + D_{p12} \hat{D}_{c2}, \qquad D_{cl22} = D_{p11} + D_{p12} \hat{D}_{c1} D_{p21}.
\end{align*}
The auxiliary controller variables are defined by
\begin{align}
\label{OF-Prf-Syn-hatABCD}
    \begin{bmatrix}
        \hat{A}_c & \hat{B}_{c2} & \hat{B}_{c1} \\
        \hat{C}_c & \hat{D}_{c2} & \hat{D}_{c1}
    \end{bmatrix} &= \begin{bmatrix}
        S A_{aug} R & S B_{aug0} X & 0 \\
        0 & 0 & 0
        \end{bmatrix} + \begin{bmatrix}
        N & S B_{aug2} \\
        0 & I
    \end{bmatrix} \begin{bmatrix}
        A_c & B_{c2} & B_{c1} \\
        C_c & D_{c2} & D_{c1}
    \end{bmatrix} \begin{bmatrix}
        M^T & 0 & 0 \\
        0 & X & 0 \\
        C_{aug2} R & 0 & I
    \end{bmatrix}.
\end{align}
Applying the Schur complement on eq. (\ref{SF-Prf-Syn-LMI3}) with respect to the third through sixth block rows and columns yields
\begin{align}
\label{temp}
\begin{bmatrix}
H_{11} + L_1 \Pi^{-1} L_{1}^T & H_{12} + L_{1} \Pi^{-1} L_2 \\
H_{21} + L_2 \Pi^{-1} L_{1}^T & H_{22} + L_{2} \Pi^{-1} L_2
\end{bmatrix} < 0,
\end{align}
where
\begin{align*}
H_{11} &= {\rm He} \{ A_{aug} R + B_{aug2} \hat{C}_c \} - \frac{d R}{d t} \\
H_{21} &= H_{12}^T = \hat{A}_c + A_{aug}^T + C_{aug2}^T \hat{D}_{c1}^T B_{aug2}^T -S \frac{d R}{d t} - N \frac{d M^T}{d t} \\
H_{22} &= {\rm He} \{ S A_{aug} + \hat{B}_{c1} C_{aug2} \} + \frac{d S}{d t} \\
L_1 & = \begin{bmatrix} B_{aug0} X + B_{aug2} \hat{D}_{c2} & B_{aug1} + B_{aug2} \hat{D}_{c1} D_{p21} & R C_{aug0}^T + \hat{C}_c^T D_{\psi1}^T & R C_{aug1}^T + \hat{C}_c^T D_{p12}^T \end{bmatrix} \\
L_2 &= \begin{bmatrix} \hat{B}_{c2} & S B_{aug1} + \hat{B}_{c1} D_{p21} & C_{aug0}^T + C_{aug2}^T \hat{D}_{c1}^T D_{\psi1}^T & C_{aug1}^T + C_{aug2}^T \hat{D}_{c1}^T D_{p12}^T
\end{bmatrix}.
\end{align*}
In eq. (\ref{temp}), The off-diagonal block $(1,2)$ can be eliminated by appropriately selecting $\hat{A}_c$ as given in (\ref{hatA}).
This choice decouples the inequality into two independent matrix conditions.
Consequently, the $(1,1)$ block depends quadratically on $\hat{C}_c$, while the $(2,2)$ block depends quadratically on  $\hat{B}_{c2}, \hat{B}_{c1}$.
Therefore, the controller variables can be obtained from the least-squares solutions of equations (\ref{hatB})-(\ref{hatC}) following the approach in \cite{Gah.Au96}.

Finally, the controller realization (\ref{OF-Thm-Syn-ABCD}) follows by inverting the relations in (\ref{OF-Prf-Syn-hatABCD}).
\end{proof}

The above reconstruction procedure shows that, although the synthesis LMIs do not explicitly involve the controller gains, a realizable LPV controller can always be recovered from the feasible solution.

\begin{remark}
The controller state matrix $A_c(\rho, \dot{\rho})$ generally depends on both $\rho$ and its derivative $\dot{\rho}$. This dependence can be eliminated by fixing $R(\rho)$ or $S(\rho)$ as constants.

The IQC-based framework provides flexibility to handle different types of delays (bounded, unbounded, continuous, discontinuous, etc.) by selecting appropriate IQC multipliers. Improving the IQC characterization directly improves synthesis accuracy and closed-loop performance.
\end{remark}

\subsection{Memoryless Control}

We further consider the case of memoryless delay control, which arises when the delay signal $\tau(t)$ is not available for feedback. In this situation, the controller cannot incorporate the delayed input information through the internal delay loop introduced in the previous subsection. As a result, the controller reduces to a standard dynamic output-feedback law that depends only on the measured output $y$. Within the proposed IQC framework, this leads to a robust stabilization problem for LPV systems with uncertain input delay using memoryless output feedback.
The LPV memoryless controller is in the form of
\begin{align}
\label{OF-Controller-memoryless}
\begin{bmatrix}
    \dot{x}_c \\
    u
\end{bmatrix} &= \begin{bmatrix}
    A_c & B_{c1} \\
    C_c & D_{c1}
\end{bmatrix} \begin{bmatrix}
    x_c \\
    y
\end{bmatrix}.
\end{align}

The corresponding synthesis condition can be obtained as a special case of Theorem \ref{OF-Thm-Syn} by enforcing $B_{c2} \equiv 0$ and $D_{c2} \equiv 0$, which effectively removes the delay loop from the controller structure. The result is summarized in the following corollary.

\begin{corollary}
\label{OF-Cor-Syn}
Consider the input-delayed LPV plant (\ref{Plant}). If there exist positive-definite matrix functions $R, S: \mathbf{R}_+ \rightarrow \mathbf{S}_+^{n_p+n_{\psi}}$, positive definite matrix functions $X_k \in \mathbf{S}_+^{n_u\times n_u}$ for all $k\in{\bf I}[1,N_\lambda]$ and a rectangular matrix function $X \in \mathbf{R}^{n_u \times n_u}$, such that conditions (\ref{OF-Cor-Syn-LMI1})-(\ref{OF-Cor-Syn-LMI3}) hold.
{\small
\begin{align}
{\cal N}_R^T \begin{bmatrix}
    \left\{ \begin{matrix} {\rm He} \{ A_{aug}(\rho) R(\rho) \} \\ -\{ \underline{\nu}, \bar{\nu} \} \frac{\partial R}{\partial \rho} \end{matrix} \right\} & \star & \star & \star & \star \\
    X^T(\rho) B_{aug0}^T(\rho) & -\displaystyle{\sum^{N_\lambda}_{k=1}} (X(\rho) + X^T(\rho) - X_k(\rho)) & \star & \star & \star \\
    B_{aug1}^T(\rho) & 0 & -\gamma I & \star & \star \\
    C_{aug0} R(\rho) & D_{\psi2} X(\rho) & 0 & -\Lambda(\rho) & \star \\
    C_{aug1}(\rho) R(\rho) & -D_{p12}(\rho) X(\rho) & D_{p11}(\rho) & 0 & -\gamma I
    \end{bmatrix} {\cal N}_R &< 0 \label{OF-Cor-Syn-LMI1} \\
\tilde{\cal N}_S^T \begin{bmatrix}
    \left\{ \begin{matrix} {\rm He} \{ S(\rho) A_{aug}(\rho) \} \\ + \{ \underline{\nu}, \bar{\nu} \} \frac{\partial S}{\partial \rho} \end{matrix} \right\} & \star & \star & \star & \star \\
    C_{aug0} & -\Lambda(\rho) & \star & \star & \star \\
    C_{aug1}(\rho) & 0 & -\gamma I & \star & \star \\
    X^T(\rho) B_{aug0}^T(\rho) S(\rho) & X^T(\rho) D_{\psi2}^T & -X^T(\rho) D_{p12}^T(\rho) &                   \left\{ \begin{matrix} -\displaystyle{\sum^{N_\lambda}_{k=1}} (X(\rho) + X^T(\rho) \\ - X_k(\rho)) \end{matrix} \right\} & \star \\
    B_{aug1}^T(\rho) S(\rho) & 0 & D_{p11}^T(\rho) & 0 & -\gamma I
    \end{bmatrix} \tilde{\cal N}_S &< 0 \label{OF-Cor-Syn-LMI2} \\
\begin{bmatrix}
    R(\rho) & I \\
    I & S(\rho)
    \end{bmatrix} &> 0 \label{OF-Cor-Syn-LMI3}
\end{align}
}
where
$\tilde{\cal N}_S = {\rm Ker} \begin{bmatrix} C_{aug2}(\rho) & 0 & 0 & 0 & D_{p21}(\rho) \end{bmatrix}$.
Then the input-delayed system (\ref{Plant}) is stabilized by a memoryless output-feedback controller
(\ref{OF-Controller-memoryless})
of the order $n_c = n_{aug}$, and achieves an $\mathcal{L}_2$ gain less than $\gamma$.
Moreover, the controller coefficient matrices $A_c, B_{c1}, C_{c}, D_{c1}$ can be recovered using the same reconstruction procedure as in Theorem \ref{thm.OF-Thm-Syn-ABCD} by setting  $B_{c2} = 0$ and $D_{c2} = 0$.
\end{corollary}

It is important to note that, under the IQC-based formulation adopted in this work, the resulting synthesis conditions for memoryless output-feedback control are not convex. In particular, the term $X^T(\rho) B_{aug0}^T(\rho) S(\rho)$ appearing in the $(4,1)$ block of (\ref{OF-Cor-Syn-LMI2}) introduces a bilinear coupling between the decision variables $X(\rho)$ and $S(\rho)$. Consequently, the synthesis conditions involve bilinear matrix inequalities (BMIs) rather than LMIs, which significantly increases the computational difficulty of the problem and generally prevents direct solution via convex optimization.

This observation underscores the importance of the proposed delay-dependent LPV controller structure, which enables convex synthesis conditions within the IQC framework, whereas the memoryless formulation inevitably leads to non-convex BMI constraints.

\section{Numerical Study}
\label{Sec.Ex}

In this example, we aim to demonstrate the advantages of the proposed IQC-based design method by exploiting information on parameter variations.
The benefits of the IQC-based delay-dependent control approach over existing LKF-based design methods have previously been demonstrated for state- and input-delayed LTI systems in terms of the maximum allowable delay bound
$\bar{\tau}$ under different delay-rate specifications
$r$ \cite{YuaW.TCYB16,YuaW.Au17}.
Here we illustrate its effectiveness in an LPV setting.

To this end, consider the following linear system with time-varying input delay:
\begin{align}
\label{Ex1.Plant}
\begin{aligned}
\dot{x}_p(t) &=
    \begin{bmatrix}
        0 & 1+\phi \rho(t) \\
        -2 & -3+\sigma \rho(t)
    \end{bmatrix} x_p(t) +
    \begin{bmatrix}
        0.1 & 0 \\ 0.1 & 0
    \end{bmatrix} d(t) + \begin{bmatrix} 0.1 + \phi \rho^2(t) \\ -0.2 + \sigma \rho(t) \end{bmatrix}
        u(t - \tau(t)) \\
        e(t) &= \begin{bmatrix} 0 & 10 \\ 0 & 0 \end{bmatrix} x_p(t) + \begin{bmatrix} 0 \\ 0.1 \end{bmatrix} u(t) \\
        y(t) &= \begin{bmatrix} 1 & 0 \end{bmatrix} x_p(t) + \begin{bmatrix} 0 & 0.05 \end{bmatrix} d(t).
\end{aligned}
\end{align}
where $\phi = 0.2, \sigma = 0.1$ and the scheduling parameter satisfies $\rho \in {\cal P} = [-5, 5]$ and the time-varying delay $\tau(t)$ satisfies $\tau(t)\in [0,\bar{\tau}]$ and $|\dot{\tau}(t)|\leq r < 2$ for all $t\geq 0$.

Under the IQC framework, the delay operator can be characterized using multipliers from \cite{KaoR.Au07}.
Since $\tau(t)$ is bounded with bounded rate of variation, we select two IQC multipliers
\begin{align}
\label{Ex1.Pi123}
\begin{aligned}
    \Pi_1(s) &= \begin{bmatrix}
        |\phi(s)|^2 & 0 \\
        0 & -1
    \end{bmatrix}, \quad
    \Pi_2(s) = \begin{bmatrix}
        |\varphi(s)|^2 & 0 \\
        0 & -1
    \end{bmatrix}
\end{aligned}
\end{align}
where
\[
\phi(s) = k_1\left(\frac{\bar{\tau}^2s^2 + c_1\bar{\tau} s}{\bar{\tau}^2s^2 + a_1\bar{\tau} s + k_1c_1}\right) + \epsilon, \qquad \varphi(s) = k_2\left(\frac{\bar{\tau}^2s^2 + c_2\bar{\tau} s}{\bar{\tau}^2s^2 + a_2\bar{\tau} s + b_2}\right) +\delta.
\]
The parameters are chosen as $k_1 = 1 + \frac{1}{\sqrt{1-r}}, a_1 = \sqrt{2k_1c_1}$, where $c_1$ is any positive number satisfying $c_1< 2k_1$; $k_2 = \sqrt{\frac{8}{2-r}}, a_2 = \sqrt{6.5 + 2b_2}, b_2 = \sqrt{50}, c_2 = \sqrt{12.5}$, and $\epsilon, \delta$ are small positive constants.
In this example we select $c_1 = 1, \epsilon = 10^{-7}, \delta = 0.001$.
It can be verified that the selected multipliers satisfy Assumption \ref{Ass234}.
Applying the $J_{n_u,n_u}$-spectral factorization \cite{Sei.TAC14,PfiS.Au15} yields the corresponding factorizations
\begin{align}
\label{Ex1.Psi12}
\begin{aligned}
    \Psi_1(s) &= \begin{bmatrix}
        \frac{\frac{k_1(c_1-a_1)}{\bar{\tau}}s - k_1^2c_1/\bar{\tau}^2}{s^2 + \frac{a_1}{\bar{\tau}}s + k_1c_1/\bar{\tau}^2} + k_1 + \epsilon & 0 \\
        0 & 1
    \end{bmatrix}, \quad
    \Psi_2(s) = \begin{bmatrix}
        \frac{\frac{k_2(c_2-a_2)}{\bar{\tau}}s - k_2b_2/\bar{\tau}^2}{s^2 + \frac{a_2}{\bar{\tau}}s + b_2/\bar{\tau}^2} + k_2 + \delta & 0 \\
        0 & 1
    \end{bmatrix}.
\end{aligned}
\end{align}
These transfer matrices can be converted into state-space form, leading to the IQC-induced dynamics in (\ref{SF-Psi-LTI}) with system matrices
\begin{align}
\label{Ex1.xpsiABCD}
\begin{aligned}
    A_\psi &=
    \begin{bmatrix}
        0 & 1 & 0 & 0 \\
        -\frac{k_1c_1}{\bar{\tau}^2} & -\frac{a_2}{\bar{\tau}} & 0 & 0 \\
        0 & 0 & 0 & 1 \\
        0 & 0 & -\frac{b_2}{\bar{\tau}^2} & -\frac{a_2}{\bar{\tau}}
    \end{bmatrix}, \quad
    B_{\psi1} = \begin{bmatrix}
        0 \\ 1 \\ 0 \\ 1
    \end{bmatrix}, \quad
    B_{\psi2} = 0, \\
    C_{\psi,1} &= \begin{bmatrix}
        -\frac{k_1^2c_1}{\bar{\tau}^2} & \frac{k_1(c_1-a_1)}{\bar{\tau}} & 0 & 0
    \end{bmatrix}, \quad
    D_{\psi1,1} = k_1 + \epsilon, \quad
    D_{\psi2,1} = 0, \\
    C_{\psi,2} &= \begin{bmatrix}
        0 & 0 & -\frac{b_2k_2}{\bar{\tau}^2} & \frac{k_2(c_2-a_2)}{\bar{\tau}}
    \end{bmatrix}, \quad
    D_{\psi1,2} = k_2+\delta, \quad
    D_{\psi2,2} = 0.
\end{aligned}
\end{align}
For $r \leq 1$, both IQC multipliers in (\ref{Ex1.Psi12}) are used for controller synthesis, whereas for $r > 1$ only $\Pi_2$ is employed.

We next compare controller synthesis using constant and parameter-dependent Lyapunov functions within the IQC framework.
Two design methods are considered:
\begin{enumerate}
\item
Delay-dependent LPV control using a quadratic Lyapunov function;
\item
Delay-dependent LPV control using parameter-dependent Lyapunov functions, with different basis selections for $R(\cdot), S(\cdot)$ and $h_1(\rho) = 1, \quad h_2(\rho) = \rho$
\begin{align*}
\mbox{Case 1}: \quad f_1(\rho) &= 1, \ f_2(\rho) = \rho, \qquad g_1(\rho) = 1 \\
\mbox{Case 2}: \quad f_1(\rho) &= 1, \qquad g_1(\rho) = 1, \ g_2(\rho) = \rho.
\end{align*}
\end{enumerate}
The parameter space ${\cal P}$ is uniformly gridded by 41 points.

The resulting closed-loop ${\cal L}_2$-gain bounds are summarized in Table \ref{Tab.gamma}. The controller synthesis problems are formulated as convex optimization programs in terms of LMIs and can therefore be solved efficiently using standard semidefinite programming solvers. The results show that parameter-dependent Lyapunov functions generally yield less conservative performance bounds than the quadratic Lyapunov function formulation, with significant  improvement when the parameter variation rate is small. This improvement stems from the ability of parameter-dependent Lyapunov functions to capture the dependence of the system dynamics on the scheduling parameter. By allowing the Lyapunov matrix to vary with $\rho$, the resulting stability certificate becomes less conservative than that obtained with a constant Lyapunov function. It is also observed that Case~1 provides better closed-loop performance than Case~2, indicating that introducing parameter dependency in $R(\cdot)$ is more effective in exploiting parameter variation information.

On the other hand, increasing either the maximum delay bound $\bar{\tau}$ or the delay derivative bound $r$ leads to larger values of $\gamma$, indicating degraded performance. Moreover, the proposed IQC-based LPV control method can handle input-delayed systems with delay derivative bound $r>1$, whereas traditional delay control methods \cite{WuG.Au01,BriSL.SCL10} fail to stabilize such systems. Therefore, the IQC-based design scheme not only provides improved performance levels across all tested combinations of $(\bar{\tau},r)$, but also significantly enhances the stability robustness of input-delayed LPV systems. Overall, this example illustrates that the proposed IQC-based LPV synthesis framework can effectively exploit parameter variation as well as allowable delay and delay-derivative information, while retaining a computationally efficient convex formulation, leading to improved stability margins and performance guarantees for LPV systems with time-varying input delays.

\begin{table}[htb]
\centering
\caption{Performance comparison of different delay dependent control methods}
\label{Tab.gamma}
\begin{tabular}{c|p{0.6in}|cccccc} \hline \hline
Method & parameter variation rate $\nu$ & \multicolumn{6}{c}{(Delay derivative $r$, delay bound $\tau$)}  \\ \hline
& & (0, 5) & (0.5, 2.5) & (0.9, 1) & (1.5, 1) & (1.7, 2.5) \\ \hline
Quadratic LF & & 7.7879 & 4.0138 & 3.1781 & 5.581 & 11.111 \\ \hline
\multirow{3}{1.0in}{Parameter-dependent LF (Case 1)}
& 0.1 & 5.013 & 2.9761 & 2.1629 & 4.6001 & 9.37 \\
& 1 & 5.0904 & 3.0049 & 2.3035 & 4.6411 & 9.4005 \\
& 10 & 6.4037 & 3.607 & 2.7689 & 5.2586 & 9.8518 \\ \hline
\multirow{3}{1.0in}{Parameter-dependent LF (Case 2)} & 0.1 & 7.2295 & 4.0124 & 3.1861 & 5.4687 & 10.907 \\
& 1 & 7.2443 & 4.0126 & 3.2093 & 5.4747 & 10.926 \\
& 10 & 7.251 & 4.0129 & 3.2616 & 5.5289 & 11.103 \\ \hline \hline
\end{tabular}
\end{table}



Time-domain simulations illustrating the resulting closed-loop transient performance will be included in the final version.

\section{Conclusions}
\label{Sec.Con}

This paper has presented an IQC-based framework for $\mathcal{H}_\infty$ output-feedback control of LPV systems with time-varying input delays. By integrating parameter-dependent Lyapunov functions with dynamic IQC representations of the delay operators, the proposed method enables the derivation of convex synthesis conditions expressed as parameter-dependent linear matrix inequalities (LMIs). The convexity of these conditions is a direct consequence of the proposed exact memory controller structure, revealing the importance of such a design in enabling systematic and tractable control synthesis for LPV systems. In contrast, for memoryless control, where the input delay is not directly measurable, the synthesis problem becomes non-convex and difficult to solve,  underscoring the value of the proposed structured controller in achieving convexity.

Numerical examples illustrate the effectiveness and advantages of the proposed approach, including enhanced closed-loop performance, improved computational tractability, and broad applicability to LPV systems with parameter-dependent dynamics. Overall, the results demonstrate that the IQC-based methodology, combined with the proposed controller structure, provides a unified, systematic, and computationally efficient approach for both analysis and synthesis of LPV systems with time-varying input delays, offering a powerful alternative to traditional Lyapunov–Krasovskii-based methods.

%

\end{document}